\documentclass{article}
\usepackage[margin=1in]{geometry}
\author{Joachim Gudmundsson, Tiancheng Mai, Sampson Wong}
\usepackage{graphicx} 
\usepackage{amsmath}
\usepackage{amsfonts}
\usepackage{amsthm}
\usepackage{xcolor}
\usepackage{enumerate}

\usepackage{hyperref}
\usepackage{cleveref}
\newtheorem{theorem}{Theorem}
\newtheorem{definition}[theorem]{Definition}
\newtheorem{lemma}[theorem]{Lemma}
\newtheorem{corollary}[theorem]{Corollary}
\newtheorem{problem}[theorem]{Problem}
\newtheorem{fact}[theorem]{Fact}
\DeclareMathOperator{\simpl}{simpl}
\DeclareMathOperator{\trough}{trough}
\newcommand{\disF}[2]{\ensuremath{d_F(#1,#2)}}
\newcommand{\dis}[3]{\ensuremath{d_{#1}(#2,#3)}}

\title{Approximating the Fr\'echet distance when only one curve is~\texorpdfstring{$c$}{c}-packed}
\date{}

\begin{document}

\maketitle

\begin{abstract}
One approach to studying the Fr\'echet distance is to consider curves that satisfy realistic assumptions. By now, the most popular realistic assumption for curves is~\mbox{$c$-packedness}. Existing algorithms for computing the Fr\'echet distance between~\mbox{$c$-packed} curves require both curves to be~\mbox{$c$-packed}. In this paper, we only require one of the two curves to be~\mbox{$c$-packed}. Our result is a nearly-linear time algorithm that $(1+\varepsilon)$-approximates the Fr\'echet distance between a~\mbox{$c$-packed} curve and a general curve in~$\mathbb R^d$, for constant values of~$\varepsilon$,~$d$ and~$c$.
\end{abstract}

\section{Introduction}

The Fr\'echet distance~\cite{Frechet1906} is a popular similarity measure between curves. The Fr\'echet distance has a variety of applications, from geographic information science~\cite{Laube2014,Ranacher2014,Toohey2015} to computational biology~\cite{Jiang2008,Wylie2013} and data mining~\cite{Kenefic2014,Wang2013}. The Fr\'echet distance can be seen as the minimum leash length of a dog walking problem.

Suppose a person and a dog walk along two polygonal curves~$P$ and~$Q$, respectively. The goal of both the person and the dog is to walk along the path, independently and at possibly different speeds, but without leaving the path or going backwards. The leash length of a given walk is defined to be the maximum distance attained between the person and the dog. The Fr\'echet distance is the globally minimum leash length over all possible walks.

The Fr\'echet distance can be computed between a pair of polygonal curves in nearly-quadratic time. Alt and Godau~\cite{Alt1995} provided an~$O(n^2 \log n)$ time exact algorithm for computing the Fr\'echet distance. Buchin, Buchin, Meulemans and Mulzer~\cite{Buchin2017} provided a randomised exact algorithm that computes the Fr\'echet distance in time~$O(n^2 \sqrt{\log n} ( \log \log n)^{3/2})$ on a pointer machine, and in time~$O(n^2 (\log \log n)^2))$ on word RAM. 

Conditional lower bounds imply that the Fr\'echet distance problem is unlikely to admit a strongly subquadratic time algorithm. Bringmann~\cite{Bringmann2014} showed that, under the Strong Exponential Time Hypothesis, the Fr\'echet distance cannot be computed in time $O(n^{2-\delta})$ for any $\delta > 0$, if we allow for approximation factors up to 1.001. Buchin, Ophelders and Speckmann~\cite{DBLP:conf/soda/BuchinOS19} showed the same conditional lower bound even if we allow for approximation factors up to 3, and even if the curves are one dimensional.

One approach to circumvent the conditional lower bounds on the Fr\'echet distance is to focus on curves that satisfy realistic assumptions. Realistic assumptions reflect the spatial distribution of curves from real-world data sets~\cite{DBLP:journals/comgeo/GudmundssonSW23}. The most popular realistic input assumption for curves under the Fr\'echet distance is~\mbox{$c$-packedness}~\cite{Driemel2012}. A curve $\pi \in \mathbb R^d$ is~\mbox{$c$-packed} if for all~$r>0$, the total length of~$\pi$ inside any ball of radius~$r$ is upper bounded by~$cr$. 

Driemel, Har-Peled and Wenk~\cite{Driemel2012} introduced the~\mbox{$c$-packedness} assumption and presented a $(1+\varepsilon)$-approximation algorithm for the Fr\'echet distance between a pair of~\mbox{$c$-packed} curves. Their algorithm runs in $O(cn/\varepsilon + cn \log n)$ time for curves in $\mathbb R^d$. Bringmann and K\"unnemann~\cite{Bringmann2017} improved the running time of the algorithm to $O(\frac{cn}{\sqrt \varepsilon} \log^2 (1/\varepsilon) + cn \log n)$ for curves in $\mathbb R^d$. Assuming the Strong Exponential Time Hypothesis, Bringmann~\cite{Bringmann2014} showed that $(i)$ for sufficiently small constants~$\varepsilon > 0$ there is no $(1+\varepsilon)$-approximation in time $O((cn)^{1-\delta})$ for any $\delta >0$, and $(ii)$ in any dimension~$d \geq 5$ there is no $(1+\varepsilon)$-approximation in time $O((cn/\sqrt \varepsilon)^{1-\delta})$ for any~$\delta > 0$. 

Existing algorithms~\cite{Bringmann2017, Driemel2012} for computing the Fr\'echet distance between~\mbox{$c$-packed} curves require that both curves are~\mbox{$c$-packed}. An open problem is whether the Fr\'echet distance can be approximated efficiently when only one curve is $c$-packed. This asymmetric case may occur if the two curves come from two different data sets. For example, in error detection we may want to match a curve containing errors to a curve close to the ground truth. 

\begin{problem}
    \label{problem:main}
    Can the Fr\'echet distance be approximated efficiently if only one of the two curves is~$c$-packed? In particular, for constant values of $\varepsilon$,~$d$ and~$c$, can we obtain a subquadratic time $(1+\varepsilon)$-approximation of the Fr\'echet distance between a~\mbox{$c$-packed} curve and a general curve in $\mathbb R^d$?
\end{problem}

We resolve Problem~\ref{problem:main} in the affirmative. Our result is an $O(c^3 (n+m) \log^{2d+1} (n) \log m)$ time algorithm that $(1+\varepsilon)$-approximates the Fr\'echet distance between a~\mbox{$c$-packed} curve with~$n$ vertices in~$\mathbb R^d$ and a general curve with~$m$ vertices in~$\mathbb R^d$, where~$\varepsilon$ is a constant. In other words, to $(1+\varepsilon)$-approximate the Fr\'echet distance in nearly-linear time, our result implies that it suffices to assume that only one of the two curves is \mbox{$c$-packed}. Our result is stated formally in Theorem~\ref{theorem:main_algorithm}. Note that for constant values of~$d$, the running time is also polynomial in~$c$ and~$\varepsilon$.

\subsection{Related work}

By now, the most popular realistic input assumption for curves under the Fr\'echet distance is $c$-packedness.  The $c$-packedness assumption has been applied to a wide variety of Fr\'echet distance problems. Typically, these algorithms incur an approximation factor of~$(1+\varepsilon)$, and have a polynomial dependence on~$\varepsilon^{-1}$. Chen, Driemel, Guibas, Nguyen and Wenk~\cite{DBLP:conf/alenex/ChenDGNW11} study the map matching problem between a $c$-packed curve and realistic graph, that is, to compute a path in the graph that is most similar to the $c$-packed curve. Har-Peled and Raichel~\cite{DBLP:journals/talg/Har-PeledR14} compute the mean curve of a set of $c$-packed curves. The mean curve is a curve that minimises its maximum weak Fr\'echet distance to the set of curves. Driemel and Har-Peled~\cite{Driemel2013} consider a variant of the Fr\'echet distance on $c$-packed curves, where any subcurve of the $c$-packed curve can be replaced by a shortcut segment. Br\"uning, Conradi and Driemel~\cite{DBLP:conf/esa/BruningCD22} and Gudmundsson, Huang, van Renssen and Wong~\cite{DBLP:conf/isaac/GudmundssonHRW23} study two distinct variants of the subtrajectory clustering problem on $c$-packed curves, that is, to detect trajectory patterns by computing clusters of subcurves. Van der Hoog, Rotenberg and Wong~\cite{DBLP:journals/corr/abs-2212-07124} study data structures for $c$-packed curves under the discrete Fr\'echet distance. Conradi, Driemel and Kolbe~\cite{DBLP:journals/cgt/ConradiDK24} consider the approximate nearest neighbour problem for $c$-packed curves in doubling metrics. Conradi, Driemel and Kolbe~\cite{DBLP:journals/corr/abs-2401-03339} compute the Fr\'echet distance between $c$-packed piecewise continuous smooth curves.

Given a polygonal curve, the problem of computing its packedness value~$c$ has been considered. Gudmundsson, Sha and Wong~\cite{DBLP:journals/comgeo/GudmundssonSW23} provide a $6.001$-approximation algorithm that runs in $O(n^{4/3} \log^9 n)$ time for curves in $\mathbb R^2$. They also provided an implementation for a $2$-approximation algorithm that runs in $O(n^2)$, and verified that $c < 50$ for a majority of data sets that were tested. Har-Peled and Zhou~\cite{DBLP:journals/corr/abs-2105-10776} provide a randomised 288.001-approximation algorithm that runs in $O(n \log^2 n)$ time and succeeds with high probability.

The $c$-packedness assumption can be applied to any set of edges, as a result, $c$-packed graphs have also been studied. Gudmundsson and Smid~\cite{DBLP:journals/comgeo/GudmundssonS15} study the map matching problem between a curve with long edges and a $c$-packed graph with long edges. They consider the data structure variant, where the graph is known in preprocessing time and the curve is only known at query time. Gudmundsson, Seybold and Wong~\cite{DBLP:journals/talg/GudmundssonSW24} generalise the result of~\cite{DBLP:journals/comgeo/GudmundssonS15} and provide a map matching data structure for any $c$-packed graph and for any query curve. 

%Lindsey's paper

\subsection{Notation}

Let $\varepsilon > 0$ be a positive real number. Without loss of generality, we can assume $0 < \varepsilon < \frac 1 2$, as providing a~$(1+\varepsilon)$-approximation for smaller values of~$\varepsilon$ also provides a~$(1+\varepsilon)$-approximation for larger values of~$\varepsilon$. 

Let $d$ be a fixed positive integer, and let $\mathbb R^d$ be $d$-dimensional Euclidean space. A polygonal curve $P = p_1 \ldots p_n$ in $\mathbb R^d$ consists of $n$ vertices $\{p_i\}_{i=1}^n$ connected by $n-1$ straight line segments~$\{p_i p_{i+1}\}_{i=1}^{n-1}$, where $p_i \in \mathbb R^d$ and~$p_i p_{i+1} \subset \mathbb R^d$. 

We define $c$-packedness. Let $c$ be a positive real number. A polygonal curve $P$ in $\mathbb R^d$ is $c$-packed if, for any radius~$r > 0$ and for any ball $B(p,r)$ centred at $p \in \mathbb R^d$ with radius~$r$, the set of segments in $P \cap B(p,r)$ has total length upper bounded by $c r$. 

Next, we define the Fr\'echet distance. Let $P = p_1 \ldots p_n$. With slight abuse of notation, define the function $P:[1,n] \to \mathbb R^d$ so that $P(i) = p_i$ for all integers $i \in \{1, \ldots, n\}$, and $P(i + x) = (1-x) p_i + x p_{i+1}$ for all reals $x \in [0,1]$. Let $\Gamma(n)$ be the space of all continuous, non-decreasing, surjective functions from $[0,1] \to [1,n]$. For a pair of polygonal curves $P = p_1,\ldots,p_n$ and $Q = q_1,\ldots,q_m$, we define the Fr\'echet distance to be

\[
    d_F(P,Q)
    =
    \inf_{\substack{\alpha \in \Gamma(n) \\ \beta \in \Gamma(m)}} \max_{\mu \in [0,1]} d(P(\alpha(\mu)), Q(\beta(\mu)))
\]
where $d(\cdot,\cdot)$ denotes the Euclidean distance in $\mathbb R^d$. 

\section{Decision algorithm}
\label{section:decision}

In this section, we solve the decision version of the Fr\'echet distance problem. We defer the optimisation version of the Fr\'echet distance problem to Section~\ref{section:optimisation}. Both the decision and optimisation versions will incur an approximation factor of $(1+\varepsilon)$.

We formally define the decision version. First we define the exact decision version. Let~$r$ be a positive real number. Let $P = p_1p_2 \ldots p_n$ be a \mbox{$c$-packed} curve in $\mathbb R^d$ and let $Q = q_1q_2 \ldots q_m$ be a general curve in $\mathbb R^d$. Given $P$, $Q$ and $r$, the exact decision problem is to answer whether $(i)$~$d_F(P,Q) \leq r$ or $(ii)$~$d_F(P,Q) > r$, where $d_F(\cdot, \cdot)$ denotes the Fr\'echet distance. Unfortunately, in our case, we will not be able to decide between~$(i)$ and~$(ii)$ exactly. Therefore, we instead solve the approximate decision version. In the approximate decision problem, we are additionally allowed a third option, that is $(iii)$~to provide a $(1+\varepsilon)$-approximation for $d_F(P,Q)$. 

We will build a decider for the approximate decision version, for any fixed $0 < \varepsilon < \frac 1 2$. Given any $P$, $Q$ and~$r$, the decider returns either $(i)$, $(ii)$ or $(iii)$. The decider requires the \mbox{$c$-packed} curve to be simplified. We will first describe the simplification procedure (Section~\ref{section:decision:simplification}), then we will construct the fuzzy decider (Section~\ref{section:decision:fuzzy}), and finally we will combine two fuzzy deciders into a complete approximate decider (Section~\ref{section:decision:complete}).

\subsection{Simplification}
\label{section:decision:simplification}

The first step in the decision algorithm is to simplify the \mbox{$c$-packed} curve~$P$. We will use the simplification algorithm in Driemel, Har-Peled and Wenk~\cite{Driemel2012}.

\begin{fact}[\cite{Driemel2012}]
\label{fact:simplification}
Given~$\mu > 0$ and a polygonal curve~$\pi = p_1 p_2 p_3 ... p_k$ in $\mathbb R^d$, we can compute in $O(k)$ time a simplification~$\simpl(\pi,\mu)$ with the following properties:
\begin{enumerate}[a)]
    \item for any vertex $p \in \pi$ there exists a vertex $q \in \simpl(\pi,\mu)$ such that $d(p,q) \leq \mu$,
    \item $d_F(\pi,\simpl(\pi,\mu)) \leq \mu$,
    \item all segments in~$\simpl(\pi,\mu)$ have length at least~$\mu$ (except the last),
    \item if $\pi$ is $c$-packed, then $\simpl(\pi,\mu)$ is $6c$-packed.
\end{enumerate}
\end{fact}
\begin{proof}
We state Algorithm 2.1 from~\cite{Driemel2012}, since we will use it in Section~\ref{section:optimisation:simplification} to determine the critical values of our algorithm. Mark the initial vertex $p_1$ and set it as the current vertex. Scan the polygonal curve from the current vertex until it reaches the first vertex $p_i$ that is at least $\mu$ away from the current vertex. Mark $p_i$ and set it as the current vertex. Repeat this until the final vertex, and mark the final vertex. Set the marked vertices to be the simplified curve, and denote it as $\simpl(\pi, \mu)$. See~\Cref{fig:simplification}. Fact 2a follows from Algorithm~2.1 in~\cite{Driemel2012}. Facts 2b, 2c and~2d follow from Lemma~2.3, Remark~2.2 and Lemma~4.3 in~\cite{Driemel2012}. 
\end{proof}

\begin{figure}[ht]
    \centering
    \includegraphics[width=0.7\textwidth]{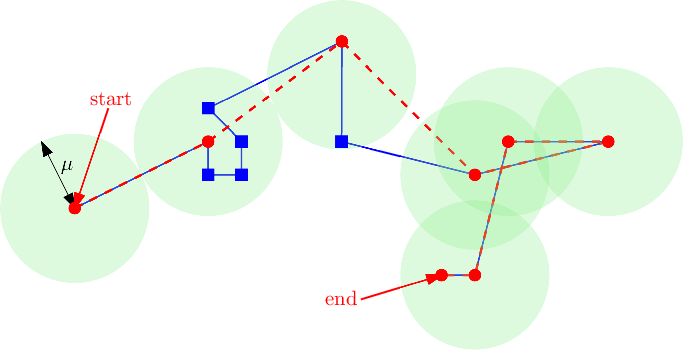}
    \caption{A polygonal trajectory $P$ (blue) and its $\mu$-simplification (red dashed). The vertices marked with blue squares are on $P$ but not included in the simplification.}
    \label{fig:simplification}
\end{figure}

\subsection{Fuzzy decider}
\label{section:decision:fuzzy}

The second step in the decision algorithm is to construct a fuzzy decider. Let $\varepsilon' = \varepsilon / 30$. Given $P$, $Q$ and~$r$, the fuzzy decision problem is to answer whether $(i)$ $d_F(P,Q) \leq (1+\varepsilon'/2)r$, or $(ii)$ $d_F(P,Q) > (1 - 2\varepsilon') r$. We call the decision problem fuzzy as there is a fuzzy region $\left( \left(1 - 2 \varepsilon' \right) r, \left(1+\varepsilon'/2 \right)r \right]$ where it would be acceptable to return either~$(i)$ or~$(ii)$. Note that unlike the complete approximate decider, for the fuzzy decider, there is no option~$(iii)$.

The overall approach in the fuzzy decider is to approximate the optimal walks along $K$ and $Q$, where~$K$ is the simplification of~$P$ from Fact~\ref{fact:simplification}. In particular, our approach is to guess how far along~$K$ we are when we reach vertex~$q_i$ on~$Q$. We use a layered directed graph to model the walk along $K$, where each layer corresponds to the walk reaching~$q_i$ on~$Q$.

The fuzzy decision algorithm constructs a layered directed graph and searches it for a suitable walk. We divide the fuzzy decision algorithm into three steps. The first step is to query a range searching data structure~\cite{Schwarzkopf1996} to construct the vertices of the graph (Section~\ref{subsubsection:vertices}). The second step is to query an approximate Fr\'echet distance data structure~\cite{Driemel2013} to construct the directed edges of the graph (Section~\ref{subsubsection:edges}). The third step is to run a breadth first search and then to return either~$(i)$ or~$(ii)$ (Section~\ref{subsubsection:returning}).

\subsubsection{Constructing the vertices}
\label{subsubsection:vertices}

The first step in the fuzzy decider is to construct the vertices of the layered directed graph. Let $\delta = \varepsilon'/2 = \varepsilon / 60$. Construct the simplification $K = \simpl(P, \delta r)$ using Fact~\ref{fact:simplification}. Recall that layer~$i$ corresponds to candidate positions on~$K$ when we reach~$q_i$ on~$Q$. Formally, define layer~$i$ to be $W_i = \{w_{i,j}\}$. All points $w_{i,j} \in W_i$ satisfy~$w_{i,j} \in K$ and~$d(w_{i,j},q_i) \leq 2r$. Note that $w_{i,j}$ is not necessarily a vertex of $P$, but rather a point on an edge of $K = \simpl(P,\delta r)$. To construct~$W_i$, we require the data structure of Schwarzkopf and Vleugels~\cite{Schwarzkopf1996}, which is a range searching data structure for low density environments.

\begin{definition}
\label{definition:low_density}
A set of objects $\Sigma$ is $k$-low density if, for every box~$H$, there are at most $k$ objects in~$H$ that intersect it and are larger than it. The size of an object is the size of its smallest enclosing box.
\end{definition}

\begin{fact}[Theorem 3 in~\cite{Schwarzkopf1996}]
\label{fact:schwarzkopf}
A $k$-low density environment $\Sigma$ of $n$ objects in $\mathbb R^d$ can be stored in a data structure of size $O(n \log^{d-1} n + kn)$, such that it takes $O(\log^{d-1} n + k)$ time to report all $E \in \Sigma$ that contains a given query point $q \in \mathbb R^d$. The data structure can be computed in $O(n \log^d n + k n \log n)$ time.
\end{fact}

We apply Fact~\ref{fact:schwarzkopf} to the curve~$K$. In particular, we turn $K$ into a low-environment in~$\mathbb R^{d+1}$ by using the trough construction of Gudmundsson, Seybold and Wong~\cite{DBLP:journals/talg/GudmundssonSW24}. The same trough construction was also used in~\cite{DBLP:conf/compgeom/BuchinBG0W24}.

\begin{lemma}
\label{lemma:lds}
Let $\delta > 0$ be fixed. Let~$P$ be a~\mbox{$c$-packed} curve with~$n$ vertices in $\mathbb R^d$. Let~$K = \simpl(P,\delta r)$. We can preprocess~$K$ into a data structure of $O(n \log^d n + c \delta^{-1} n)$ size, so that given a query point $q \in \mathbb R^d$, the data structure can return in $O(\log^d n + c \delta^{-1})$ time all~$O(c \delta^{-1})$ edges of~$K$ that are within a distance of $2r$ from~$q$. The preprocessing time is $O(n \log^{d+1}n + c \delta^{-1} n \log n)$.
\end{lemma}

\begin{proof}
The curve~$K$ is $6c$-packed by Fact~\ref{fact:simplification}d. Next, we generalise the trough construction of Gudmundsson, Seybold and Wong~\cite{DBLP:journals/talg/GudmundssonSW24} to $(d+1)$-dimensions. 
We define a trough object in $\mathbb R^{d+1}$ for every segment~$e \in K$ by $\trough(e, \delta) = \{\left(x_1,\ldots,x_d,z\right): d\left(\left(x_1,\ldots,x_d\right),e\right) \leq 4z \leq 8 \delta^{-1} |e|\}$, where $d(\cdot,\cdot)$ and $|\cdot|$ are measured under the Euclidean metric in $\mathbb R^d$. Let $T$ be the set of all trough objects. By Lemma~23 in~\cite{DBLP:journals/comgeo/GudmundssonSW23},~$T$ is an $O(c \delta^{-1})$-low-density environment. We apply the data structure from Fact~\ref{fact:schwarzkopf} on the environment~$T$.  

Given a query point $q = (x_1,\ldots,x_d)$, we query the data structure for all troughs that contain the $(d+1)$-dimensional point $(x_1,\ldots,x_d,r/2)$. Suppose the data structure returns a set of $k$ objects $\{\trough(e_{i}, \delta) \}^k_{i=1}$. Then $k = O(c \delta^{-1})$, since $T$ is an $O(c \delta^{-1})$-low-density environment, and $q$ has zero size. From the set of~$k$ troughs we extract the set of~$k$ edges $\{e_{i}\}^k_{i=1}$.

The running times follow from Fact~\ref{fact:schwarzkopf} and from $T$ being an $O(c \delta^{-1})$-low-density environment. It remains to prove the correctness of the query. Recall the definition of the trough that $(x_1,\ldots,x_d,r/2) \in \trough(e, \delta)$ if and only if $d\left(\left(x_1,\ldots,x_d\right),e\right) \leq 4\cdot \frac{r}{2} \leq 8\delta^{-1}|e|$.
In particular, $d\left(q,e\right) \leq 2r$ covers all edges in~$K$ that intersect a ball of radius $2r$ centred at $q$ and $4\delta^{-1}|e| \geq r$ covers all edges of length at least $\delta r/4$. Since $K$ is also $(\delta r)$-simplified, all edges of $K$ (except for the last edge) are at least of length $\delta r$ by Fact~\ref{fact:simplification}c. We can check the last edge of~$K$ separately. 
\end{proof}

We use Lemma~\ref{lemma:lds} to construct~$W_i$ for all $1 \leq i \leq m$. Recall that $Q = q_1 \ldots q_m$. Query the data structure in Lemma~\ref{lemma:lds} to obtain all edges in $K = \simpl(P,\delta r)$ that are within a distance of $2r$ from $q_i$. Let this set of edges be $T_i$. Note that $|T_i| = O(c \delta^{-1})$, since $K$ is $6c$-packed and each edge in~$T_i$ has length at least~$\delta r$. For each edge $e_{i, j} \in T_i$, we choose $O(\delta^{-1})$ evenly spaced points on the chord $e_{i,j} \cap B(q_i,2r)$, so that the distance between two consecutive points on the chord is less than $\delta r$. We add these evenly spaced points to $W_i$ for each $e_{i,j} \in T_i$, so that in total, $|W_i| = O(c \delta^{-2})$. See \Cref{fig:fuzzy-decider-local}.

\begin{figure}[ht]
    \centering
    \includegraphics[width=0.7\textwidth]{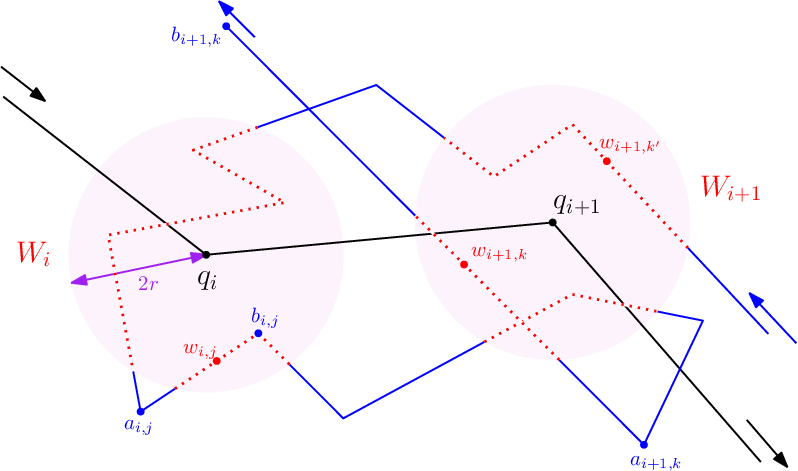}
    \caption{The general curve $Q$ (black), the $(\delta r)$-simplification $K$ (blue) and two candidate sets $W_i$ and $W_{i+1}$ (red dots). The coloured arrows indicate the order of vertices on the curve.
    A candidate set $W_i$ (red dots) contains evenly spaced points on $K$ chords that are at most a distance $2r$ away from $q_i$, i.e., in the violet shading.
    The point $w_{i,j}$ is on the edge $a_{i,j}b_{i,j}$ and the point $w_{i+1,k}$ is on the edge $a_{i+1,k}b_{i+1,k}$.
    }
    \label{fig:fuzzy-decider-local}
\end{figure}

This completes the construction of~$W_i$ for $1 \leq i \leq m$.  Since $q_1, q_{m}$ must be matched to $p_1, p_n$ respectively, we can simplify the sets $W_1 = \{p_1\}$ and $W_m = \{p_n\}$. The vertices of our graph are $\cup_{i=1}^m W_i$, which completes the first step of the construction of the fuzzy decider.

\subsubsection{Constructing the edges}
\label{subsubsection:edges}

The second step in the fuzzy decider is to construct the edges of the layered directed graph. Each edge in the graph is a directed edge from $W_i$ to $W_{i+1}$ for some $1 \leq i \leq m-1$. A directed edge from $w_{i,j} \in W_i$ to $w_{i+1,k} \in W_{i+1}$ models a simultaneous walk, from $w_{i,j}$ to $w_{i+1,k}$ and from $q_i$ to~$q_{i+1}$, on~$K$ and~$Q$ respectively. We only add this directed edge into the graph if its associated walk is feasible. To decide whether the walk is feasible, we check two conditions. The first condition is that~$w_{i,j}$ preceeds~$w_{i+1,k}$ along the curve $K$. The second condition is whether the Fr\'echet distance between the subcurve~$K \langle w_{i,j}, w_{i+1,k} \rangle$ and the segment~$q_i q_{i+1}$ is at most~$r$. See \Cref{fig:overview-matching}. To efficiently check the second condition, we require the approximate Fr\'echet distance data structure of Driemel and Har-Peled~\cite{Driemel2013}. 

\begin{figure}[ht]
    \centering
    \includegraphics[width=0.5\textwidth]{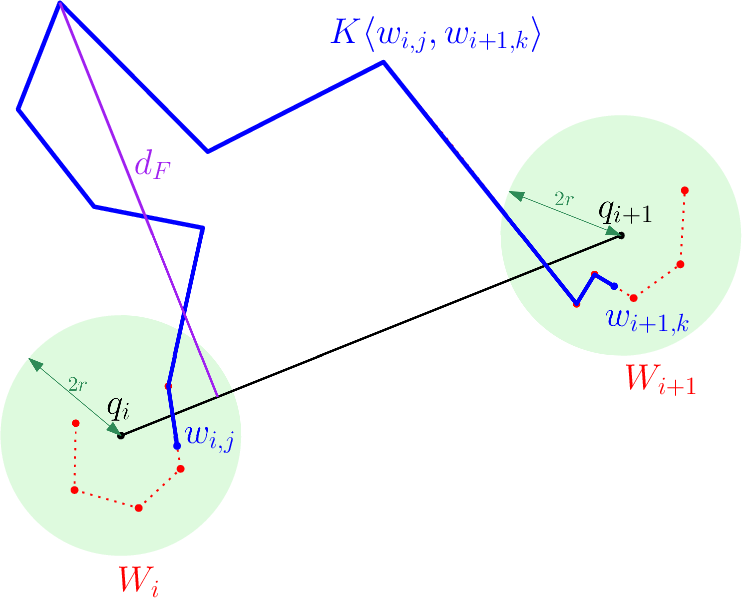}
    \caption{The Fr{\'e}chet distance (purple), between a segment $(q_i, q_{i+1})$ (black) and subcurve $K \langle w_{i, j}, w_{i+1, k} \rangle $ (blue). A candidate set $W_i$ (red dots) contains evenly spaced points on $K$ chords that are at most a distance $2r$ away from $q_i$, i.e., in the green shading.}
    \label{fig:overview-matching}
\end{figure}

\begin{fact}[Theorem 5.9 in \cite{Driemel2013}]
\label{fact:a-segment-query-to-a-subcurve}
Given $\delta > 0$ and a polygonal curve $Z$ with $n$ vertices in $\mathbb{R}^d$, one can construct a data structure in $O(\delta^{-2d} \log^2(1/\delta) n \log^2 n)$ time that uses $O(\delta^{-2d} \log^2(1/\delta) n)$ space, such that for a query segment $pq$, and any two points $u$ and $v$ on the curve, one can $(1 + \delta)$-approximate the distance $d_F(Z \langle u,v \rangle , pq)$ in $O(\delta^{-2} \log n \log \log n)$ query time. 
\end{fact}

We construct the data structure in Fact~\ref{fact:a-segment-query-to-a-subcurve} on the curve~$K$. Let $1 \leq i \leq m-1$, $w_{i,j} \in W_i$ and $w_{i+1,k} \in W_{i+1}$. We query the data structure in Fact~\ref{fact:a-segment-query-to-a-subcurve} to compute a $(1+\delta)$-approximation of $d_F(K \langle w_{i,j}, w_{i+1,k} \rangle, q_i q_{i+1})$. If the reported value is at most~$r$, then we insert the directed edge from $w_{i,j}$ to $w_{i+1,k}$. We repeat this for all $1 \leq i \leq m-1$, $w_{i,j} \in W_i$ and $w_{i+1,k} \in W_{i+1}$. This completes the construction of the edges in the directed graph, and completes the second step of the fuzzy decider.

\subsubsection{Returning either \texorpdfstring{$(i)$}{(i)} or \texorpdfstring{$(ii)$}{(ii)}}
\label{subsubsection:returning}

The third step of the fuzzy decider is to run a breadth first search on the layered directed graph. Recall that $W_1 = \{p_1\}$ and $W_m = \{p_n\}$. We use the breadth first search to decide whether there is a directed path from $p_1$ to $p_n$. Recall that $\varepsilon' = \varepsilon/30$ and $\delta = \varepsilon / 60$. If there is a directed path, we return $(i)$~$d_F(P,Q) \leq (1 + \varepsilon'/2) r$. Otherwise, we return $(ii)$~$d_F(P,Q) > (1 - 2\varepsilon') r$. Next, we prove the correctness of the fuzzy decider. We have two cases.

\begin{itemize}
    \item There is a directed path from $p_1$ to $p_n$ in the layered directed graph. Let the directed path be $c_1 \ldots c_m$. Then $c_i \in W_i$ for all $1 \leq i \leq m$. We match the vertex $q_i$ to $c_i$ for all $1 \leq i \leq m$. We match the segment $q_i q_{i+1}$ to the subcurve $K\langle c_i, c_{i+1} \rangle$ for all $1 \leq i \leq m-1$. Since there is a directed edge from $c_i$ to $c_{i+1}$, we have that the estimated Fr\'echet distance between the segment $q_i q_{i+1}$ and the subcurve $\disF{K \langle c_{i},c_{i+1} \rangle }{q_{i}q_{i+1}}$ is at most $r$. Formally, we have $C_i \leq r$, where
    \[
          \disF{K \langle c_{i},c_{i+1} \rangle }{q_{i}q_{i+1}} \leq C_i \leq (1+\delta) \cdot \disF{K \langle c_{i},c_{i+1} \rangle }{q_{i}q_{i+1}}.
    \]
   In particular, we have $\disF{K \langle c_{i},c_{i+1} \rangle }{q_{i}q_{i+1}} \leq r$. Taking the maximum over all $1 \leq i \leq m-1$, we get
    \[
        \disF{K}{Q} \leq \max_{i=1, \ldots, m-1} \disF{K \langle c_{i},c_{i+1} \rangle }{q_{i}q_{i+1}} \leq  r.
    \]
    By Fact~\ref{fact:simplification}b, we have $\disF{P}{K} \leq \delta r$. Since the Fr{\'e}chet distance obeys the triangle inequality, we have 
    \[
        \disF{P}{Q} \leq \disF{P}{K}  + \disF{K}{Q} \leq  r + \delta r \leq (1+\varepsilon'/2) r. 
    \]
    Therefore, it is correct to return $(i)$~$d_F(P,Q) \leq (1 + \varepsilon'/2) r$ in the case where there is a directed path from $p_1$ to $p_n$.
    
    \item There is no directed path from $p_1$ to $p_n$ in the layered directed graph. Let $r^{*} = \disF{P}{Q}$. Suppose that for the optimal Fr\'echet distance between $P$ and $Q$, we match $q_i \in Q$ to $p_i^{*} \in P$ for all $1 \leq i \leq m$. Therefore $\dis{}{q_i}{p_i^{*}} \leq r^{*}$. Let $r' = \disF{K}{Q}$. Suppose that for the optimal Fr\'echet distance between~$K$ and~$Q$, we match $q_i \in Q$ to ${k}_i^{*} \in K$ for all $1 \leq i \leq m$. Therefore $\dis{}{q_i}{{k}_i^{*}} \leq r'$. Since the Fr{\'e}chet distance obeys the triangle inequality, we have $r' = \disF{K}{Q} \leq \disF{P}{Q} + \disF{K}{P} = r^{*} + \delta r$.

    Assume for the sake of contradiction that $r^{*} \leq (1 - 2 \varepsilon'
    ) r$.
    Then, we have \[r' \leq r^{*} + \delta r \leq \left(\delta + 1 - 2 \varepsilon'
    \right) r <  r < 2r.\]
    Therefore, $\dis{}{q_i}{{k}_i^{*}} \leq r' < 2r$.
    Thus, there exists ${k}_i^{*} \in K$ that is at most a distance $2r$ away from $q_i$ and the edge that ${k}_i^{*}$ resides on is also at most $2r$ away from $q_i$. Hence, $W_i$ is non-empty, and there exists $u_i \in W_i$ such that $u_i$ and ${k}_i^{*}$ share the same chord (edge) in $K$ and $\dis{K}{u_i}{{k}_i^{*}} \leq \delta r$. In particular, there exists $u_{i},v_{i} \in W_i$ where $k_i^*$ is on the subcurve~$K \langle u_i, v_i \rangle$, so that $\dis{K}{u_{i}}{{k}_i^{*}} \leq \delta r$, $\dis{K}{v_{i}}{{k}_i^{*}} \leq \delta r$, and $\dis{K}{u_{i}}{v_{i}} \leq \delta r$. See~\Cref{fig:break_r_prime_i}.
    
    \begin{figure}[ht]
        \centering
        \includegraphics[width=0.7\textwidth]{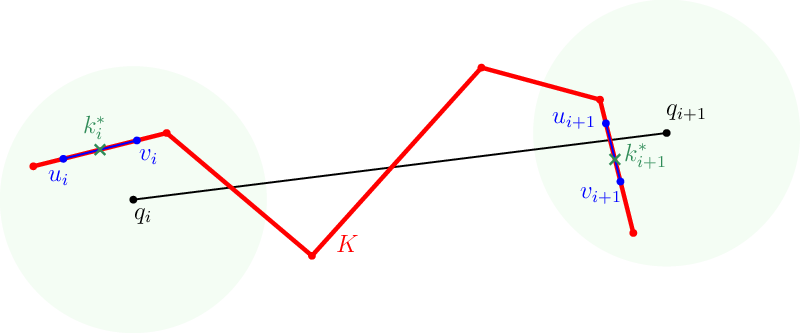}
        \caption{The point $k^{*}_i$ marked with a green cross, and its immediate neighbour points $u_i$ and $v_i$, marked with blue dots. Note that $k^{*}_i$ is on the subcurve~$K \langle u_i, v_i \rangle$, so that $\dis{K}{u_{i}}{{k}_i^{*}} \leq \delta r$, $\dis{K}{v_{i}}{{k}_i^{*}} \leq \delta r$, and $\dis{K}{u_{i}}{v_{i}} \leq \delta r$.}
        \label{fig:break_r_prime_i}
    \end{figure}
    
    Consider $r_i' = \disF{{K}\langle u_{i}, u_{i+1}\rangle }{q_{i}q_{i+1}}$ when $q_{i} \in Q$ is matched to $u_i \in W_i \subset K$ and $q_{i+1} \in Q$ is matched to $u_{i+1} \in W_{i+1} \subset K$. Then
    \begin{align*}
        r_i' 
        &\leq
        \disF{{K}\langle {k}_{i}^{*}, {k}_{i+1}^{*}\rangle }{q_{i}q_{i+1}} +
        \disF{{K}\langle {k}_{i}^{*}, {k}_{i+1}^{*}\rangle }{{K}\langle u_{i}, u_{i+1}\rangle } \\
        &\leq
        r' +
        \disF{{K}\langle {k}_{i}^{*}, {k}_{i+1}^{*}\rangle }{{K}\langle u_{i}, u_{i+1}\rangle } \\
        &\leq r' +
        \disF{{k}_{i}^{*} \circ {K}\langle {k}_{i}^{*}, u_{i+1}\rangle  \circ u_{i+1} {k}_{i+1}^{*} }
        {u_{i}{k}_{i}^{*} \circ {K}\langle {k}_{i}^{*}, u_{i+1}\rangle  \circ u_{i+1}} \\
        &\leq r' + \max \left\{
            \disF{{k}_{i}^{*}}{u_{i}{k}_{i}^{*}},
            \disF{{K}\langle {k}_{i}^{*}, u_{i+1}\rangle }{{K}\langle {k}_{i}^{*}, u_{i+1}\rangle },
            \disF{u_{i+1} {k}_{i+1}^{*}}{u_{i+1}}
        \right\} \\
        &\leq r' + \max \left\{
            \dis{K}{{k}_{i}^{*}}{u_{i}},
            0,
            \dis{K}{{k}_{i+1}^{*}}{u_{i+1}}
        \right\} \\
        &\leq r' + \max \left\{
            \delta r,
            0,
            \delta r
        \right\} \\
        &\leq r' + \delta r
    \end{align*}
    where $\circ$ denotes the concatenation of polygonal curves. Therefore, 
    \[
    r_i' \leq r' + \delta r \leq \left(\delta +
    1 - 2 \varepsilon' + \delta
    \right) r = \left(2\delta +
    1 - 2 \varepsilon'
    \right) r \leq \left(
    1 - \varepsilon'
    \right) r.
    \]
    
    Let $C_i$ be the $(1+\delta)$-approximation of $\disF{{K}\langle u_{i}, u_{i+1}\rangle }{q_{i}q_{i+1}}$ returned by the data structure in Fact~\ref{fact:a-segment-query-to-a-subcurve}. Then, $r_i' \leq C_i \leq (1+\delta) r_i'$. Therefore, $C_i \leq (1+\delta) r_i' <  (1+\varepsilon') r_i' \leq (1+\varepsilon') (
    1 - \varepsilon'
    ) r < r$ for all $1 \leq i \leq m$. Hence, there is a directed edge from $u_i$ to $u_{i+1}$ in the layered directed graph, for all $1 \leq i \leq m-1$. In particular, $u_1 \ldots u_m$ is a directed path from $p_1$ to $p_n$, which is a contradiction. We conclude that our assumption $r^{*} \leq (
    1 - 2 \varepsilon'
    ) r$ cannot hold, and it is correct to return $r^{*} > (
    1 - 2 \varepsilon'
    ) r$ in the case where there is no directed path from $p_1$ to $p_n$.
\end{itemize}

We obtain the following theorem.

\begin{theorem}[Fuzzy decider]
\label{theorem:fuzzy_decider}
Given a positive real number $r$, $0 < \varepsilon < \frac 1 2$, and a $c$-packed curve $P$ with $n$ vertices in~$\mathbb R^d$, one can construct a data structure in $O(n \log^{d+1} n + c \varepsilon^{-1} n \log n + \varepsilon^{-2d} \log^2(1/\varepsilon) n \log^2 n)$ time that uses $O(n \log^d n + c \varepsilon^{-1} n + \varepsilon^{-2d} \log^2 (1/\varepsilon) n)$ space, so that given a query curve~$Q$ with~$m$ vertices, the data structure returns in $O(\log^d n + m c^2 \varepsilon^{-6} \log n \log \log n)$ query time either $(i)$~$d_F(P,Q) \leq (1 + \varepsilon'/2) r$ or $(ii)$~$d_F(P,Q) > (1 - 2 \varepsilon') r$.
\end{theorem}

\begin{proof}
First, we summarise the preprocessing procedure. Let $\delta = \varepsilon/60$. We use Fact~\ref{fact:simplification} to construct the simplification~$K = \simpl(P, \delta r)$. We use Lemma~\ref{lemma:lds} to construct a range searching data structure on~$K$, and we use Fact~\ref{fact:a-segment-query-to-a-subcurve} to construct an approximate distance data structure on~$K$. Next, we summarise the query procedure. We query Lemma~\ref{lemma:lds} to construct $W_i$ for $1 \leq i \leq m$, we query Fact~\ref{fact:a-segment-query-to-a-subcurve} to construct the edges between $W_i$ and $W_{i+1}$ for $1 \leq i \leq m-1$, and finally we run a breadth first search. We argued correctness in Section~\ref{subsubsection:returning}. It remains to analyse preprocessing time, space, and query time.

The preprocessing time of Fact~\ref{fact:simplification}, Lemma~\ref{lemma:lds} and Fact~\ref{fact:a-segment-query-to-a-subcurve} is $O(n \log^{d+1} n + c \delta^{-1} n \log n + \delta^{-2d} \log^2(1/\delta) n \log^2 n)$. The space of the data structures in Lemma~\ref{lemma:lds} and Fact~\ref{fact:a-segment-query-to-a-subcurve} is $O(n \log^d n + c \delta^{-1} n + \delta^{-2d} \log^2 (1/\delta) n)$. Substituting $\delta^{-1} = O(\varepsilon^{-1})$ yields the stated preprocessing time and space. 

We analyse the query time. Constructing the set~$W_i$ for all $1 \leq i \leq m$ takes $O(m(\log^d n + c \delta^{-2}))$ time, since using Lemma~\ref{lemma:lds} to query the set of edges close to~$q_i$ takes $O(\log^d n + c \delta^{-1})$ time, and constructing evenly spaced points takes $O(c \delta^{-2})$ time. Since $|W_i| = O(c \delta^{-2})$, the number of pairs $\cup_{i=1}^{m-1} (W_i \times W_{i+1})$ is $O(m c^2 \delta^{-4})$. Querying Fact~\ref{fact:a-segment-query-to-a-subcurve} to decide whether there is a directed edge takes $O(\delta^{-2} \log n \log \log n)$ time per pair. In total, constructing the edges in the layered directed graph takes $O(m c^2 \delta^{-6} \log n \log \log n)$ time. Running breadth first search takes $O(m c^2 \delta^{-4})$ time. The total query time is $O(\log^d n + c^2 \delta^{-6} m \log n \log \log n)$. Substituting $\delta^{-1} = O(\varepsilon^{-1})$ yields the stated query time. 
\end{proof}

\subsection{Complete approximate decider}
\label{section:decision:complete}

The third step in the decision algorithm is to use the fuzzy decider to construct a complete approximate decider. Recall that, given $\varepsilon$, $P$, $Q$ and $r$, the complete approximate decider returns either $(i)$~$d_F(P,Q) \leq r$, $(ii)$~$d_F(P,Q) > r$, or $(iii)$~a $(1+\varepsilon)$-approximation for $d_F(P,Q)$. 

\begin{theorem}[Complete approximate decider]
\label{theorem:complete_decider}
Given a positive real number $r$, $0 < \varepsilon < \frac 1 2$, and a $c$-packed curve $P$ with $n$ vertices in~$\mathbb R^d$, one can construct a data structure in $O(n \log^{d+1} n + c \varepsilon^{-1} n \log n + \varepsilon^{-2d} \log^2(1/\varepsilon) n \log^2 n)$ time that uses $O(n \log^d n + c \varepsilon^{-1} n + \varepsilon^{-2d} \log^2 (1/\varepsilon) n)$ space, so that given a query curve~$Q$ with~$m$ vertices in~$\mathbb R^d$, the data structure returns in $O(\log^d n + m c^2 \varepsilon^{-6} \log n \log \log n)$ query time either $(i)$~$d_F(P,Q) \leq r$, $(ii)$~$d_F(P,Q) > r$, or $(iii)$~a $(1+\varepsilon)$-approximation for $d_F(P,Q)$.
\end{theorem}

\begin{proof}
Let~$r_1 = \frac 1 {1 + \varepsilon'/2} \, r$ and~$r_2 = \frac 1 {1 - 2\varepsilon'} \, r$, where~$\varepsilon' = \varepsilon/30$. First, given $r_1$,~$\varepsilon$, and~$P$, construct the data structure in Theorem~\ref{theorem:fuzzy_decider}, and query the data structure on~$Q$. Second, given $r_2$,~$\varepsilon$, and~$P$, construct the data structure in Theorem~\ref{theorem:fuzzy_decider}, and query the data structure on~$Q$. If the first query returns~$(i)$, we return~$(i)$. If both the first and second queries return~$(ii)$, we return~$(ii)$. Otherwise, if the first query returns~$(ii)$ and the second query returns~$(i)$, we return~$(iii)$. We prove correctness in three cases.

\begin{itemize}
    \item The first query returns~$(i)$. Then by Theorem~\ref{theorem:fuzzy_decider}, $d_F(P,Q) \leq (1+\varepsilon'/2) r_1 = r$, so returning~$(i)$ in the complete approximate decider is correct.
    \item Both the first and second queries return~$(ii)$. Then by Theorem~\ref{theorem:fuzzy_decider}, $d_F(P,Q) > (1-2\varepsilon') r_2 = r$, so returning~$(ii)$ in the complete approximate decider is correct.
    \item The first query returns~$(ii)$ and the second query returns~$(i)$. The first query implies $d_F(P,Q) >
    (1-2 \varepsilon') \cdot r_1
    = (1-2 \varepsilon') \cdot \frac{1}{1 + \varepsilon'/2} \cdot r 
    $. The second query implies $d_F(P,Q) \leq ( 1 + \varepsilon'/2 ) \cdot r_1 = ( 1 + \varepsilon'/2 ) \cdot \frac{1}{1-2 \varepsilon'} \cdot r$. Putting these together, we have \[d_F(P,Q) \in
    \left(
    \frac{1-2 \varepsilon'
    }{1 + \varepsilon'/2 }  r,
    \frac{1 + \varepsilon'/2 }{1-2 \varepsilon'
    }  r
    \right].
    \]
    Note that  
    \[
    \frac{\frac{1 + \varepsilon'/2 }{1-2 \varepsilon'}  r}{\frac{1-2 \varepsilon'}{1 + \varepsilon'/2 }  r}
    =
    \left(\frac{1 + \varepsilon'/2 }{1-2 \varepsilon'}\right)^2 
    <
    \left((1+\varepsilon'/2) (1+4\varepsilon')\right)^2 
    <
    (1+6 \varepsilon')^2 
    <
    1 + 30 \varepsilon' 
    =
    1 + \varepsilon.
    \]
    Hence, $ 
    \frac{1-2 \varepsilon'
    }{1 + \varepsilon'/2 }  r$ is a $(1 + \varepsilon)$-approximation of $d_F(P,Q)$, so returning~$(iii)$ in the complete approximate decider is correct.
\end{itemize}
Finally, the preprocessing time, space, and query time follow from Theorem~\ref{theorem:fuzzy_decider}.
\end{proof}

This completes the decision version of the approximate Fr\'echet distance problem. Next, we consider the optimisation version of the approximate Fr\'echet distance problem.

\section{Optimisation algorithm}
\label{section:optimisation}

In Section~\ref{section:optimisation:simplification}, we apply a binary search to compute the optimal simplification. In Section~\ref{section:optimisation:parametric search}, we apply parametric search to compute the Fr\'echet distance. In both steps, we use the complete approximate decider in Theorem~\ref{theorem:complete_decider}, which incurs an approximation factor of~$(1+\varepsilon)$.

\subsection{Approximating the optimal simplification}
\label{section:optimisation:simplification}

First, we provide an algorithm to compute the optimal simplification of~$P$. In particular, the optimal simplification is $K^* = \simpl(P,\delta r^*)$, where $\delta = \varepsilon / 60$ and $r^* = d_F(P,Q)$. Our approach is to search over the critical values of the $\mu$-simplification algorithm in Fact~\ref{fact:simplification}. A critical value of the $\mu$-simplification algorithm is a value of~$\mu$ where the simplification changes. Define the set of pairwise distances of~$P$ to be~$L(P) = \{d(p_i, p_j): 1 \leq i < j \leq n\}$. We can observe that the set of pairwise distances~$L$ breaks up the positive real line into~$\binom{n}{2} + 1$ intervals, such that within each interval the $\mu$-simplification does not change. This observation follows from the algorithm in Fact~\ref{fact:simplification}, and the same observation is made in Section~3.3.3 in~\cite{Driemel2012}. Unfortunately, $|L| = O(n^2)$. To overcome this, we use approximate distance selection. 

\begin{fact}[Lemma~3.9 in~\cite{Driemel2012}]
\label{fact:wspd}
    Given a set~$P$ of $n$ points in~$\mathbb R^d$, one can compute in $O(n \log n)$ time a set~$Z$ of~$O(n)$ numbers, such that for any $y \in L(P)$, there exists numbers $x, x' \in Z$ such that $x \leq y \leq x' \leq 2 x$. 
\end{fact}

We can refine Fact~\ref{fact:wspd} to obtain Corollary~\ref{corollary:wspd}. We replace the 8-WSPD in Lemma~3.9 of~\cite{Driemel2012} with an $8/\varepsilon$-WSPD. 

\begin{corollary}
\label{corollary:wspd}
    Given a set~$P$ of $n$ points in~$\mathbb R^d$, one can compute in $O(n / \varepsilon^d + n \log n)$ time a set~$Z$ of~$O(n / \varepsilon^d)$ numbers, such that for any $y \in L(P)$, there exists numbers $x, x' \in Z$ such that $x \leq y \leq x' \leq (1+\varepsilon) x$. 
\end{corollary}

Next, we perform binary search on the set $Z$ in Corollary~\ref{corollary:wspd}. In particular, for $x \in Z$, we decide whether~$\delta r^* < x$ or~$\delta r^* > x$ by running the complete approximate decider in Theorem~\ref{theorem:complete_decider} on $r = x/\delta$, $\delta = \varepsilon / 60$, $P$ and~$Q$. After $O(\log n)$ applications of the complete approximate decider, we obtain $\delta r^* \in [x,x']$ for a consecutive pair of elements $x,x' \in Z$. We have two cases. In the first case, we compute the optimal simplification of~$P$, that is, $K^* = \simpl(P,\delta r^*)$. In the second case, we compute a $(1+\varepsilon)$-approximation of~$r^* = d_F(P,Q)$. 

\begin{itemize}
    \item If $x' > (1+\varepsilon) x$. By the contrapositive of Corollary~\ref{corollary:wspd}, there is no $y \in L(P) \cap [x,x']$. In other words, within the interval~$[x,x']$ the simplification of~$P$ does not change. Therefore, $K^* = \simpl(P,x) = \simpl(P,\delta r^*)$.
    \item If $x' \leq (1+\varepsilon) x$. Therefore, $x'/\delta$ is a $(1+\varepsilon)$-approximation of~$r^* = d_F(P,Q)$, as required.
\end{itemize}

Therefore, we can compute $K^* = \simpl(P, \delta r^*)$, as otherwise we would have a $(1+\varepsilon)$-approximation of~$r^*$. The running time is dominated by the $O(\log n)$ applications of the complete approximate decider.

\subsection{Approximating the Fr\'echet distance}
\label{section:optimisation:parametric search}

From Section~\ref{section:optimisation:simplification}, we computed the simplification $K^* = \simpl(P,\delta r^*)$. Let~$r_1^* = \frac 1 {1 + \varepsilon'/2} \, r^*$,~$r_2^* = \frac 1 {1 - 2\varepsilon'} \, r^*$. We can use the same procedure to compute the simplifications~$K_1^* = \simpl(P,\delta r_1^*)$ and~$K_2^* = \simpl(P, \delta r_2^*)$. If $K_1^* \neq K^*$, then there must be an element $x \in Z$ in the interval~$[\delta r_1^*, \delta r^*]$, so $x/\delta$ would be a~$(1+\varepsilon)$-approxmation of~$r^*$. Therefore, $K^* = K_1^*$, and similarly, $K^* = K_2^*$.

We proceed with parametric search. Note that in Section~\ref{section:optimisation:simplification}, we did not apply parametric search to compute~$K^*$ due to efficiency reasons. It is not straightforward to parallelise Fact~\ref{fact:simplification}, moreover, since the simplification~$K^* = K_1^* = K_2^*$ does not change during the execution of the parametric search, we can avoid reconstructing the data structures in Lemma~\ref{lemma:lds} and Fact~\ref{fact:a-segment-query-to-a-subcurve}. We obtain the following theorem.

\begin{theorem}
\label{theorem:main_algorithm}
Given $\varepsilon > 0$, a $c$-packed curve~$P$ in~$\mathbb R^d$, and a general curve~$Q$ in~$\mathbb R^d$, one can compute a~$(1+\varepsilon)$-approximation of~$d_F(P,Q)$ in $O(T_s T_p \log m)$ time, where 
\begin{align*}
    T_s &= n \log^{d+1} n + c \varepsilon^{-1} n \log n + \varepsilon^{-2d} \log^2(1/\varepsilon) n \log^2 n + mc^2 \varepsilon^{-6} \log n \log \log n, \\
    T_p &= \log^d n + c \varepsilon^{-1} + \varepsilon^{-2} \log n \log \log n.
\end{align*}
\end{theorem}

\begin{proof}
First, we summarise the preprocessing procedure. We compute the simplification $K^* = \simpl(P, \delta r^*)$ using the procedure described in Section~\ref{section:optimisation:simplification}. We build the data structures in Lemma~\ref{lemma:lds} and Fact~\ref{fact:a-segment-query-to-a-subcurve} on the simplified curve~$K^*$. 

Second, we summarise the query procedure. Here, we use parametric search. We use the algorithm in Theorem~\ref{theorem:complete_decider} as both the decision algorithm and the simulated algorithm. We describe the simulated algorithm. Let~$r$ be the search parameter. Let~$r_1 = \frac 1 {1 + \varepsilon'/2} \, r$ and~$r_2 = \frac 1 {1 - 2\varepsilon'} \, r$. We simulate the complete approximate decider by simulating the fuzzy decider in Theorem~\ref{theorem:fuzzy_decider} on~$r_1$ and~$r_2$. We divide the simulation of the fuzzy decider on~$r_1$ into three steps. First, we compute~$W_i$ by querying the data structure in Lemma~\ref{lemma:lds}. We use parametric search and the decision algorithm (Theorem~\ref{theorem:complete_decider}) to resolve the critical values in the query. Second, we compute the directed edges from~$W_i$ to~$W_{i+1}$ by querying the data structure in Fact~\ref{fact:a-segment-query-to-a-subcurve}. We apply parametric search in the same way. Third, we run a breadth first search on the layered directed graph. There are no critical values in this step, so we do not need to apply parametric search. We repeat the simulation of the fuzzy decider on~$r_2$. Finally, by parametric search, we return the optimal value~$r^*$.

Third, we argue correctness. If Theorem~\ref{theorem:complete_decider} returns~$(iii)$ at any point, we obtain a~$(1+\varepsilon)$-approximation of~$r^*$, and we are done. If Theorem~\ref{theorem:complete_decider} never returns~$(iii)$ at any point, we will show that the decision algorithm and the simulated algorithm are both correct. The decision algorithm is correct since we either return~$r^* \leq r$ or $r^* > r$. We show the preprocessing and query procedures of the simulated algorithm are correct. In particular, we will show that we correctly simulate the execution of Theorem~\ref{theorem:complete_decider} as though~$r = r^*$. The preprocessing procedure is correct, since~$K^* = K_1^* = K_2^*$, so our data structures are correct for~$r_1^*$ and~$r_2^*$. The query procedure is correct, since we can use the correct decision algorithm to resolve all critical values, and simulate the correct execution path as though~$r = r^*$. Moreover, Theorem~\ref{theorem:complete_decider} (without~$(iii)$) acts discontinuously at~$r = r^*$, so $r^*$ is a critical value of the simulated algorithm. Therefore, parametric search is able to locate~$r^*$ and return it.

Fourth, we analyse the running time. The preprocessing time is dominated by $O(\log n)$ calls to Theorem~\ref{theorem:complete_decider}. The query time is dominated by parametric search. The running time of parametric search is $O(P_p T_p + T_p T_s \log P_p)$, where $T_s$ is the sequential running time of the decision algorithm, $P_p$ is the number of processors used in the simulated algorithm, and~$T_p$ is the number of parallel steps used by the simulated algorithm. The sequential running time is~$T_s = O(n \log^{d+1} n + c \varepsilon^{-1} n \log n + \varepsilon^{-2d} \log^2(1/\varepsilon) n \log^2 n + mc^2 \varepsilon^{-6} \log n \log \log n)$ by Theorem~\ref{theorem:complete_decider}. The simulated algorithm can be efficiently parallelised. In particular, the simulated algorithm computes~$W_i$ by querying Lemma~\ref{lemma:lds}, and computes the directed edges from~$W_i$ to~$W_{i+1}$ by querying Lemma~\ref{fact:a-segment-query-to-a-subcurve}; these can be queried in parallel for all~$1 \leq i \leq m$. Given~$P_p = m$ processors, we can perform all of these queries in in $T_p = O(\log^d n + c \delta^{-1} + \delta^{-2} \log n \log \log n)$ parallel steps. The overall running time is dominated by~$O(T_s T_p \log m)$, which the stated running time.
\end{proof}

We can simplify the running time if $\varepsilon$ is constant.

\begin{corollary}
Given a constant $\varepsilon > 0$, a $c$-packed curve~$P$ with~$n$ vertices in~$\mathbb R^d$, and a general curve~$Q$ with~$m$ vertices in~$\mathbb R^d$, one can~$(1+\varepsilon)$-approximate~$d_F(P,Q)$ in $O(c^3 (n+m) \log^{2d+1} (n) \log m)$ time.
\end{corollary}

\section{Conclusion}

In this paper, we provide an~$O(c^3 (n+m) \log^{2d+1} (n) \log m)$ time algorithm to~$(1+\varepsilon)$-approximate the Fr\'echet distance between two curves in $\mathbb R^d$, in the case when only one curve is $c$-packed and~$\varepsilon$ is constant. The running time is nearly-linear if~$c$ and~$d$ are also constant. An open problem is whether the running time can be improved, in particular, whether the dependence on $\varepsilon$, $c$, $d$, $\log n$ or $\log m$ can be reduced. Another open problem is whether we can obtain results for related problems when only one of the two curves is $c$-packed. Yet another open problem is whether similar results can be obtained for other realistic input curves. In particular, can the Fr\'echet distance be~$(1+\varepsilon)$-approximated in subquadratic time when only one of the curves is $\kappa$-bounded, or when only one of the curves is $\phi$-low density?

\bibliographystyle{plain}
\bibliography{bib}

\end{document}